\documentclass[copyright,creativecommons]{eptcs}
\providecommand{\event}{MeCBIC 2012} 
\usepackage{breakurl}             

\usepackage{amsfonts,amssymb,amsthm,amsmath}
\usepackage{algorithmic,algorithm} 
\usepackage[pst-pdf=off]{gastex} 
\usepackage{graphicx,subfig} 
\usepackage[figuresright]{rotating} 
\usepackage{xspace}
\usepackage{multirow}

\newcommand{\picscale}{0.8} 

\newtheorem{theorem}{Theorem}
\newtheorem{lemma}[theorem]{Lemma}
\newtheorem{definition}{Definition}

\usepackage{color} 
\usepackage[normalem]{ulem} 
\usepackage{cancel} 

\newcommand{\myskip}[0]{\medskip}
\newcommand{\mi}[1]{\mathit{#1}}

\newcommand{\Nat}{{\mathbb N}}

\newcommand{\powset}[1]{\mathcal{P}({#1})} 
\newcommand{\set}[2]{\{ {#1} \ | \ {#2}\}} 
\newcommand{\AU}[2]{A[#1 \ U \ #2]} 
\newcommand{\AUw}[2]{A[#1 \ U_{w} \ #2]} 
\newcommand{\non}[0]{\neg} 
\newcommand{\limpl}[0]{\rightarrow} 
\newcommand{\sat}[0]{\vDash} 
\newcommand{\notsat}[0]{\not\sat} 
\newcommand{\har}[0]{{\upharpoonright}} 
\newcommand{\shar}[0]{{\lceil}} 
\newcommand{\reaction}[2]{ #1 \; \rightarrow \; #2}
\newcommand{\reactionC}[3]{\reaction{#1}{#2} \;\; \{\; #3 \;\} }

\newcommand{\irule}[2]{\dfrac{\,\textstyle\rule[-1.3ex]{0cm}{3ex}#1\,}%
{\textstyle\rule[-.5ex]{0cm}{3ex}#2}}

\newcommand{\ie}{i.e.\xspace}

\newcommand{\etal}{et al.\xspace}

\def\fair{\Phi}
\def\fairabs{\Phi_{\alpha}}
\def\satfair{\sat_{\fair}}
\def\satfairabs{\sat_{\fairabs}}
\def\satLTL{\sat_{LTL}}
\def\ACTL{\ensuremath{\text{ACTL}^{-}}\xspace}
\def\emptypath{\epsilon}
\def\emptypathinf{\emptypath^{\infty}}
\def\proj{\har}
\def\sproj{\shar}
\def\sprojinf{\sproj^{\infty}}
\def\implies{\limpl}
\newcommand{\trans}[3]{{#1}\xrightarrow{#2}{#3}}
\def\Sem{\ensuremath{\mathit{LTS}}\xspace}
\def\Semabs{\ensuremath{\mathit{LTS_{\alpha}}}\xspace}
\def\re{\ensuremath{\mathit{re}}\xspace}
\def\pro{\ensuremath{\mathit{pro}}\xspace}
\def\cat{\ensuremath{\mathit{cat}}\xspace}
\def\species{\ensuremath{\mathit{species}\xspace}}
\def\comp{\ensuremath{\mathit{comp}}\xspace}
\def\pathway{\ensuremath{P_{\mathit{EGF}}}\xspace} 
\def\GF{\ensuremath{\mathit{GF}}\xspace}
\def\enabled{\ensuremath{\mathit{enabled}}\xspace}
\def\AP{\ensuremath{\mathit{AP}}\xspace}
\def\true{\ensuremath{\mathit{true}}\xspace}
\def\false{\ensuremath{\mathit{false}}\xspace}
\def\occR{\ensuremath{\mi{occurred}(R)}\xspace} 
\newcommand{\reactionRate}[3]{ #1 \; \xrightarrow{#3} \; #2}
\newcommand{\reactionCRate}[4]{\reactionRate{#1}{#2}{#4} \;\; \{\; #3 \;\} }

\def\statefont{\mathbf}
\newcommand{\state}[1]{\ensuremath{\statefont{s_{#1}}}\xspace}
\def\speciesfont{\mathrm}
\mathchardef\mhyphen="2D
\def\ERKPP{\ensuremath{\speciesfont{ERK\mhyphen PP}}\xspace}
\def\ERKPPi{\ensuremath{\speciesfont{ERK\mhyphen PPi}}\xspace}
\def\specI{\ensuremath{\speciesfont{(EGF\mhyphen EGFR^*)2\mhyphen GAP}}\xspace}
\def\Rafstar{\ensuremath{\speciesfont{Raf^*}}\xspace}
\def\compfont{\mathit}

\title{Towards modular verification of pathways:\\fairness and assumptions}
\author{
Peter Dr\'abik
\institute{Istituto di Informatica e Telematica\\
Consiglio Nazionale delle Ricerche\\
Pisa, Italy}
\email{peter.drabik@iit.cnr.it}
\and Andrea Maggiolo-Schettini
\institute{Dipartimento di Informatica\\
Universit\`{a} di Pisa\\
Pisa, Italy}
\email{maggiolo@di.unipi.it}
\and Paolo Milazzo
\institute{Dipartimento di Informatica\\
Universit\`{a} di Pisa\\
Pisa, Italy}
\email{milazzo@di.unipi.it}
}
\def\titlerunning{Towards modular verification of pathways: fairness and assumptions}
\def\authorrunning{P.~Dr\'abik, A.~Maggiolo-Schettini and P.~Milazzo}
\begin{document}
\maketitle

\begin{abstract}
Modular verification is a technique used to face the state explosion problem often encountered in the verification of properties of complex systems such as concurrent interactive systems. The modular approach is based on the observation that properties of interest often concern a rather small portion of the system. As a consequence, reduced models can be constructed which approximate the overall system behaviour thus allowing more efficient verification.

Biochemical pathways can be seen as complex concurrent interactive systems. Consequently, verification of their properties is often computationally very expensive and could take advantage of the modular approach.

In this paper we report preliminary results on the development of a modular verification framework for biochemical pathways. We view biochemical pathways  as concurrent systems of reactions competing for molecular resources. A modular verification technique could be based on reduced models containing only reactions involving molecular resources of interest.

For a proper description of the system behaviour we argue that it is essential to consider a suitable notion of fairness, which is a well-established notion in concurrency theory but novel in the field of pathway modelling. We propose a modelling approach that includes fairness and we identify the assumptions under which verification of properties can be done in a modular way.

We prove the correctness of the approach and demonstrate it on the model of the EGF receptor-induced MAP kinase cascade by Schoeberl \etal
\end{abstract}

\section{Introduction}

A big challenge of current biology is understanding the principles and functioning of complex biological systems. Despite the great effort of molecular biologists investigating the functioning of cellular components and networks, we still cannot provide a detailed answer to the question ``how a cell works?''.

In the last decades, scientists have gathered an enormous amount of molecular level information. To uncover the principles of functioning of a biological system, just collecting data does not suffice. Actually, it is necessary to understand the functioning of parts and the way these interact in complex systems. The aim of \emph{systems biology} is to build, on top of the data, the science that deals with principles of operation of biological systems. The comprehension of these principles is done by modelling and analysis exploiting mathematical means. 

A typical scenario of modelling a biological system is as follows. To build a model that explains the behaviour of a real biological system, first a formalism needs to be chosen. Then a model of the system is created, simulation is performed, and the behaviour is observed. The model is validated by comparing the results with the real experiments. The advantage of simulation is not only validation of laboratory experiments, but also prediction of behaviour under new conditions and automation of the whole process.

Simulation can give either the average system behaviour or a number of possible system behaviours. This may be insufficient when one is interested in analysing all the behaviours of a system.

Model checking may be of help. This technique permits the verification of properties (expressed as logical formulae) by exploring all the possible behaviours of a system. This analysis technique typically relies on a state space representation whose size, unfortunately, makes the analysis often intractable for realistic models. This is true in particular for systems of interest in systems biology (such as metabolic pathways, signalling pathways, and gene regulatory networks), which often consist of a huge number of components interacting in different ways, thus exhibiting very complex behaviours. 

Many formalisms originally developed by computer scientists to model systems of interacting components have been applied to biology, also with extensions to allow more precise descriptions of the biological behaviours \cite{1231124,Cardelli:2005fk,1570750,1041036,513277,citeulike:1180143}. Examples of well-established formal frameworks that can be used to model, simulate and model check descriptions of biological systems are \cite{1570750,Fages04modellingand,1342619}.

Model checking techniques have traditionally suffered from the state explosion problem. 
Standard approaches to the solution of this problem are based on abstractions or similar model reduction techniques (e.g.~\cite{186051}). Moreover, the use of Binary Decision Diagrams (BDDs) \cite{Clarke99} to represent the state space (symbolic model checking) often allows significantly larger model to be treated~\cite{DBLP:journals/iandc/BurchCMDH92}.

A method for trying to avoid the state space explosion problem is to consider a decomposition of the system, and to apply a modular verification technique allowing global properties to be inferred from properties of the system components.
This approach can be particularly efficient when the modelled systems consist of a high number of components, whereas properties of interest deal only with rather small subset of them. This is often the case for properties of biological systems.
Hence, for each property it would be useful to be able to isolate a minimal fragment of the model that is necessary for verifying such a property. If such a fragment can be obtained by working only on the syntax of the model, the application of a standard verification technique on the semantics of the fragment avoids the state explosion.

In previous work we developed a modular verification technique in which the system of interest is described by means of a general automata-based formalism suitable for qualitative description of a large class of biological systems, called sync-programs, which supports modular construction \cite{DrabikENTCS10,DrabikSACS11}. Sync-programs include a notion of synchronization that enables the modelling of biological systems. The modular verification technique is based on property preservation and allows the verification of properties expressed in the temporal logic \ACTL to be verified on fragments of models. In order to handle modelling and verification of more realistic biological scenarios, we have proposed a dynamic version of our formalism along with an extension of the modular verification framework \cite{DrabikNCMA10}.

The long-term aim of our research is the development of an efficient modular verification framework specifically designed for biochemical pathways, and of a pathway analysis tool based on such a framework. Presently, we are at the first stages of the development of the modular verification framework. However, we already faced some problems whose solution required the definition of concepts related to the formal modelling of biochemical pathways and that we believe could be interesting not only in the context of modular verification. In particular, we defined a notion of \emph{fairness} for biochemical pathways and a notion of \emph{molecular component} of a pathway. The former is a well-known concept in concurrency theory that could be useful to describe more accurately the dynamics of a pathway (in a qualitative framework). The latter is a notion relating species involved in the same pathway such that two species are considered to be part of the same molecular component if they can be seen as different 
states of the same 
molecule. As far as we know, the adoption of a notion of fairness in the context of biology is new. On the other hand, the notion of molecular component has been often implicitly used (for instance in the modelling of biological systems by means of automata), but now we provide new insight on this notion.

In this paper we report preliminary results obtained during the development of the modular verification framework. Modular verification requires either adopting a modular notation for pathway modelling or finding a way to decompose a pathway, simply expressed as a set of biochemical reactions, into a number of modules. The approach that we choose to follow is in between these two alternatives. Actually, we assume the pathway to be expressed as a set of reactions satisfying some modularisation requirements, and then we define a modularisation procedure that allows modules to be inferred from reactions. Modules will be molecular components, hence our modularisation procedure will allow us to consider a pathway not only as a set of reactions, but also as a set of entities interacting with each other (through reactions) and consequently changing state.

Once the molecular components of a pathway are identified, we can use them to decompose the verification of a global pathway property into the verification of a number of sub-properties related with groups of components. To this aim we define a \emph{projection operation} that allows a model fragment describing the behaviour of a group of components to be obtained from a model describing the whole pathway. Such a projection operation is actually an abstraction function, since the behaviour of the group of components will be over-approximated (\ie the model will include behaviours that are not present in the model of the whole pathway). By considering a suitable temporal logic for the specification of properties (namely \ACTL, a fragment of the CTL logic consisting only of universally quantified formulae) we can prove that properties holding in model fragments obtained by projection also hold in the complete model of the pathway. Nothing can be said, instead, of properties that do not hold in the the model 
fragment.

In order to verify properties of complete pathway models or of model fragments it is possible to translate them into the input language of an existing model checking tool. Specifically, we use the NuSMV model checker \cite{CAV02}, which is a well-established and efficient instrument.

We demonstrate the modular verification approach on the model of the EGF receptor-induced MAP kinase cascade by Schoeberl \etal \cite{egf-reduced} and we discuss how we plan to continue the development of the approach to improve its efficiency.

\section{Modelling Biochemical Pathways with a Notion of Fairness}

In biochemistry, metabolic pathways are networks of biochemical reactions occurring within a cell. The reactions are connected by their intermediates: products of one reaction are substrates for subsequent reactions. Reactions are influenced by catalysts and inhibitors, which are molecules (proteins) which can stimulate and block the occurrence of reactions, respectively. For the sake of simplicity we do not consider inhibitors in this paper, although they could be easily dealt with.

\subsection{Syntax and semantics of the modelling notation}
\label{sec:syntax}

Given an infinite set of species $S$, let us assume biochemical reactions constituting a pathway to have the following form:
\[
 \reactionC{r_1,\ldots,r_n}{p_1,\ldots,p_{n'}}{c_1,\ldots,c_m}
\]
where $r_j,p_j$ and $c_j$, for suitable values of $j$, are all in $S$.
We have that $r_j$s are reactants, $p_j$s are products and $c_j$s are catalysts of the considered reaction.
Given a reaction $R$ we define $\re(R) = \{r_1,\ldots,r_n\}$, $\pro(R) = \{p_1,\ldots,p_{n'}\}$, and $\cat(R) = \{c_1,\ldots,c_m\}$. We denote the set of species involved in reaction $R$ as $\species(R) = \re(R) \cup \pro(R) \cup \cat(R)$.

A pathway $P$ is simply a set of reactions, $P=\{R_{1},\hdots,R_{N}\}$. Given a pathway $P$, we can infer the set of species involved in it as $\species(P) = \bigcup_{R\in P} \species(R)$.

The dynamics of a pathway can be described at several different levels of abstraction. The most precise level consists of a quantitative description in which quantities (or concentrations) of species are taken into account, as well as reaction rates in either a deterministic or a stochastic framework. At a more abstract level reaction rates can be ignored. Ultimately, also quantities of species can be ignored by considering only their presence (or absence) in the considered biochemical solution. The less abstract description level is obviously the most precise, but also the most difficult to treat with formal analysis techniques. The more abstract levels are more suitable for the application of formal analysis techniques and are often precise enough to provide some information on the role of the species and of the reactions involved in the pathway. We choose to adopt the most abstract description level, and hence we define a qualitative formal semantics of pathways in which species can only be either present 
or absent.

The dynamics of a pathway starts from an initial state representing a biochemical solution and is determined by the reactions. A reaction essentially causes the appearance of some new species in the biochemical solution. Actually, we choose to interpret the effect of a reaction depending on whether it is catalysed or not. In our interpretation a reaction without catalysts creates the products but does not consume the reactants. We choose this interpretation since non-catalysed reactions usually reach a steady-state of dynamic equilibrium in which both reactants and products are present in the biochemical solution. On the other hand, a reaction favoured by catalysts usually tends to be performed as long as there are reactants. Therefore, in our interpretation a reaction with catalysts creates the products and consumes the reactants. 
This choice implies that a reversible reaction in which both directions are catalysed, which frequently occurs in biological pathways, oscillates between two states. This is realistic in some cases (oscillatory behaviours) but not always. We leave a more detailed treatment of this aspect as future work.

Lastly, we assume that all of the catalysts are required to be present in order for the reaction to occur. Alternative combinations of catalysts that may enable the reaction should be modelled as different reactions having the same reactants and products.

Formally, given a pathway $P$ and a set $\state{0}\subseteq \species(P)$ representing species present in the initial state of the system, the \emph{semantics} of $P$ is given by the labelled transition system 
$(\powset{\species(P)},\state{0},\to_{R})$, where $\powset{\species(P)}$ is the powerset of the set of species of $P$, meaning that each state of the LTS is a configuration of the pathway indicating which species are present. We use the boldface notation, e.g.~\state{} to denote states in the semantics, while a simple $s$ denotes a species which means either a reactant, a product or a catalyst.
Furthermore, $\to_{R}:\powset{\species(P)}\times P\times\powset{\species(P)}$
is the least transition relation satisfying the following inference
rules
\begin{gather*}
\irule{
\re(R)\subseteq \state{},\ \pro(R)\not\subseteq \state{},\ \emptyset \neq \cat(R) \subseteq \state
}{
\trans{\state{}}{R}{(\state{} \setminus \re(R))\cup \pro(R)}
}(\mbox{cat})
\qquad
\irule{
\re(R)\subseteq \state{},\ \pro(R)\not\subseteq \state{},\ \cat(R)=\emptyset
}{
\trans{\state{}}{R}{\state{} \cup pro(R)}
}(\mbox{no-cat}).
\end{gather*}
Rules (cat) and (no-cat) formalise the dynamics of reactions in the presence and absence of catalysts, respectively. Both rules contain an assumption which states that the reaction does not occur if its products already exist. Note that thanks to this optimisation transitions that do not change the state of the system are excluded, which is convenient for the verification as the size of the transition system is smaller but the set of properties that hold stays the same.
We denote the semantic function as $\Sem$, \ie $\Sem:P\mapsto (\powset{\species(P)},\state{0},\to_{R})$.

A path in $\Sem(P)$ can be either a finite sequence $\state{0},R_0,\state{1},R_1,\hdots,\state{n}$
or an infinite sequence $\state{0},R_0,$ $\state{1},R_1,\hdots$ where for all $i$, $\state{i}$ is a state and $R_i$ is a reaction and $\trans{\state{i}}{R_i}{\state{i+1}}$ is a transition in $\Sem(P)$. The path consisting only of the initial state $\state{0}$ is denoted $\emptypath$.
In this paper we consider only maximal paths, corresponding to behaviours of the pathway in which as long as some reactions can occur, the pathway activity does not halt. It is worth noting that maximal paths are not necessarily infinite, as a state where no reactions can occur has no successor and a path leading to such a state is finite.

\subsection{Fairness}

In order to describe the behaviour of a pathway more accurately we consider a notion of fairness. We motivate it by considering a quantitative system consisting of four reactions $\reactionCRate{A}{B}{D}{k_{1}}$,   $\reactionCRate{B}{A}{D}{k_{2}}$,  $\reactionCRate{A}{C}{D}{k_{3}}$ and  $\reactionCRate{C}{A}{D}{k_{4}}$, where $k_{1}$, $k_{2}$, $k_{3}$ and $k_{4}$ are the reaction rates. By performing the qualitative abstraction, we get a pathway containing reactions $R_{1}=\reactionC{A}{B}{D}$ and $R_{2}=\reactionC{B}{A}{D}$, $R_{3}=\reactionC{A}{C}{D}$ and $R_{4}=\reactionC{C}{A}{D}$, whose semantics as defined above includes behaviours such as the one where $R_{3}$ never occurs.
Such a behaviour is a qualitative abstraction which is not correct, since the standard quantitative dynamics ruled by the law of mass action would imply that both $R_{1}$ and $R_{3}$ occur with a frequency proportional to their kinetic rates.
Actually, in a stochastic setting both $R_{1}$ and $R_{3}$ would infinitely occur with probability 1.
A correct qualitative abstraction of our system should therefore only include maximal paths in which both $R_{1}$ and $R_{3}$ occur infinitely many times.

A concept from concurrency theory that allows to specify the correct behaviour is fairness, which stipulates that reactions should compete in a fair manner.
We consider the well-known notion of strong fairness \cite{EL87}, also called compassion, which requires that if a reaction is enabled (ready to occur) infinitely many times, then it will occur infinitely many times.

Technically, fairness is specified by a linear temporal logic (LTL) formula.
LTL \cite{ltl} is built up from formulae over a finite set of atomic propositions $S$, therefore a $s\in S$ is a LTL formula and if $f$ and $g$ are LTL formulae, then so are $\neg f$, $f \lor g$, $X\ g$ and $f\ U\ g$ where $X$ is read as next and $U$ as until. Additional logical operators can be defined, $\true=s \lor \neg s$, $\false=\neg \true$, $f\land g=\neg(\neg f \lor \neg g)$ and $f\limpl g=\neg f \lor g$, additional temporal operators eventually $F\ g=\true\ U\ g$ and globally $G\ g=\neg F\ \neg g$. A LTL formula can be satisfied by a maximal path $\pi$ in $\Sem(P)$ as described by the satisfaction relation $\satLTL$: $\pi \satLTL s$ if $\pi=\state{},R,\pi'$, $\pi \satLTL \neg g$ if not $\pi \satLTL g$, $\pi \satLTL f \lor g$ if $\pi \satLTL f$ or $\pi \satLTL g$, $\pi \satLTL X\ g$ if $\pi=\state{},R,\pi'$ and $\pi' \satLTL g$, and finally $\pi \satLTL f\ U\ g$ if there is an $i\geq 0$ such that $\pi=\state{0},R_{0},\state{1},R_{1},\hdots$ 
and $\state{i},R_{i},\pi_{i}\satLTL f$ and forall $0\leq k<i$, $s_{k},R_{k},\pi_{k}\satLTL g$.

Fairness is expressed by formula $\fair$, and a maximal path $\pi$ is fair iff it satisfies $\fair$, \ie $\pi \satLTL \fair$.
We have 
$$\fair \iff \bigwedge_{R\in P} (\GF\ \enabled(R)\limpl \GF\ \occR)$$ 
where $\enabled(R) \iff ( (\bigwedge_{r \in \re(R)}r) \land (\bigvee_{p \in \pro(R)}\neg p) \land (\bigwedge_{c \in \cat(R)}c) )$ and the satisfaction of proposition $\occR$ is defined as $\pi' \satLTL \occR$ iff there is a path $\pi$ such that $\pi=\state{},R,\pi'$.

It should be noted that our fairness neither requires all reactions to occur infinitely nor requires fair paths to be infinite.

\subsection{Modelling the EGF receptor-induced MAP kinase cascade}

We apply our modular verification approach to a well-established computational model of the EGF signalling pathway. We consider the model of the MAP kinase cascade activated by surface and internalised EGF receptors, proposed by Schoeberl et al.~in \cite{egf-reduced}. This model includes a detailed description of the reactions that involve active EGF receptors and several effectors named GAP, ShC, SOS, Grb2, RasGDP/GTP and Raf. Moreover, the model describes the activity of internalised receptors, namely receptors that are no longer located on the cell membrane, but on a vesicle obtained by endocytosis and floating in the cytoplasm. Such internalised receptors continue to interact with effectors and to contribute to the pathway functioning, but actually the pathway can be seen as composed by two almost identical branches: the first consisting of the reactions stimulated by receptors on the cell membrane and the second consisting of reactions stimulated by internalised receptors.

\begin{figure} 
  \includegraphics[width=17cm]{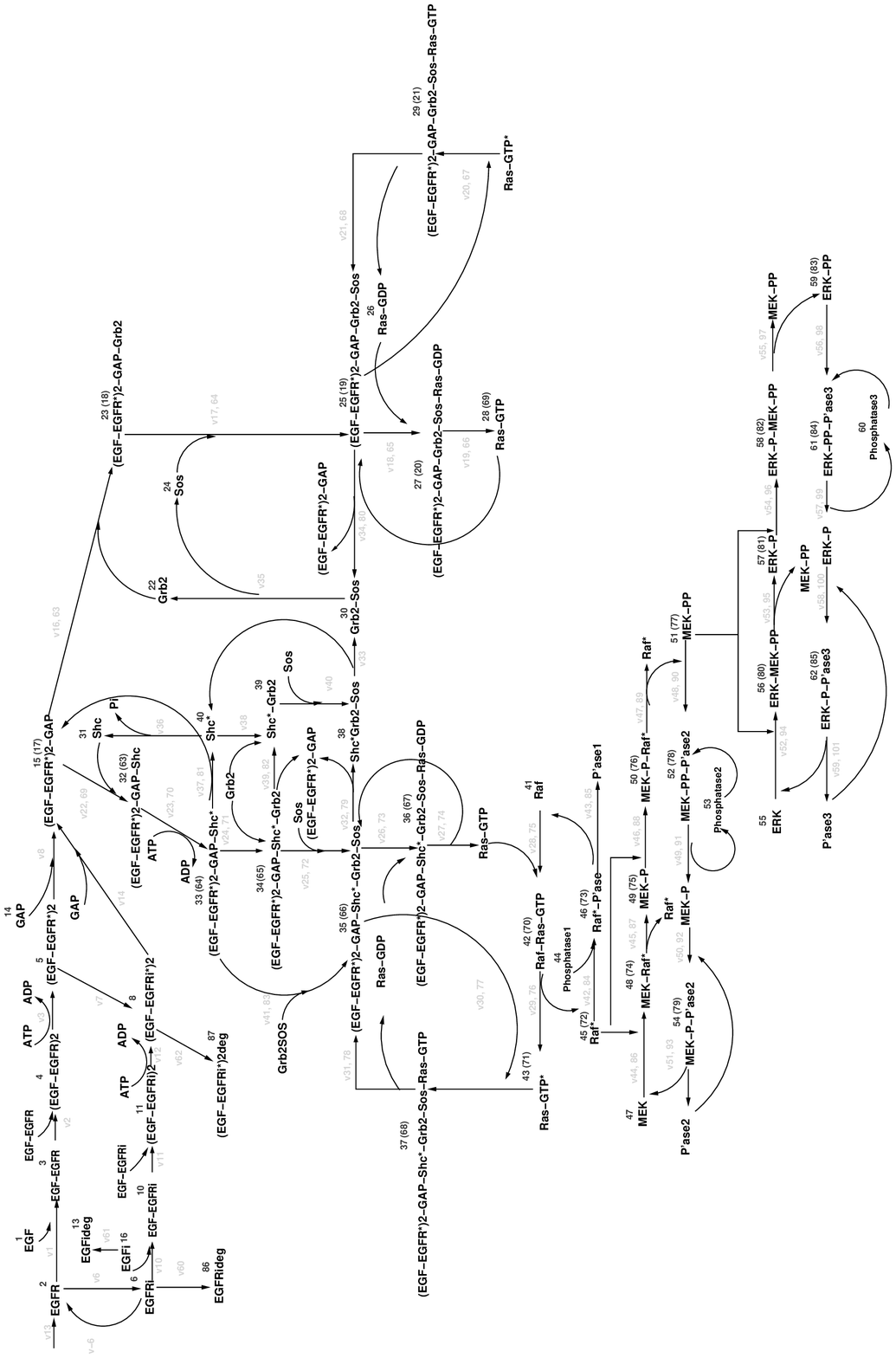}
  \caption{Scheme of the EGF receptor-induced MAP kinase cascade \cite{egf-reduced}}
  \vspace{5cm}
  \label{fig:egf-pathway}
\end{figure}

A diagram representing all of the reactions of the pathway considered in the model is shown in Figure \ref{fig:egf-pathway}. In the figure, species are identified by a short name, but also by a number (in black) in the interval $[1-60]$. Arrows represent reactions, which are also associated with an identifier (in grey) in the interval $[v0-v101]$. Note that the two branches of the pathway are partially combined in the figure. In particular, the representation of most of the species is combined with the representation of its internalised counterpart. In such cases, the number between brackets denotes the number identifying the internalised species. The same holds for reactions: in many cases an arrow denotes both a reaction stimulated by receptors in the cell membrane and the corresponding reaction stimulated by internalised receptors.

The set of reactions constituting the pathway can be trivially reconstructed from the diagram in Figure \ref{fig:egf-pathway}. The only non-trivial aspect is related with the presence in the diagram of some reactions in which one reactants is actually acting as a catalyst. For instance, this happens in the case of the reactions involving Raf$^*$ and MEK, in which Raf$^*$ initially binds MEK and then releases it phosphorylated. We describe these two reactions in the diagram with the following single catalysed reaction:
\[
  \reactionC{\speciesfont{MEK}}{\speciesfont{MEK\mhyphen P}}{\speciesfont{Raf^*}}
\]

Other species acting as catalysts are MEK-PP, Phosphatase1, Phosphatase2 and Phosphatase3. By applying the same transformation also to the reactions they are involved in we obtain a pathway constituted by 80 reactions. We call this pathway $\pathway$.

We recall that fairness requires that a reaction that is infinitely often enabled is also infinitely often performed. This prevents starvation situations to happen among reactions. In the case of $\pathway$ the two branches of the pathway include reactions that could be involved in infinite loops (e.g.~the reactions involving MEK and ERK). This means that the semantics of the pathway includes behaviours in which only one branch executes forever even if the other is constantly enabled. Such unrealistic behaviours are excluded by the adoption of fairness.

\section{Identification of Molecular Components}

In this section we argue, that under conditions often found in practice, a pathway can be decomposed into components, which, as it will be shown in the following sections, can be used for modular verification.

\subsection{Assumptions}
\label{sec:modver-assump}

Intuitively, a species can be seen as a part of  a ``state'' or ``configuration'' of a more general system component, and a reaction can be seen as a synchronised state change of a set of such system components. 
In order to view a pathway through this optics, it is convenient to assume that the pathway has equal number of reactants and products (which is not the case in general). Moreover, we assume a positional correspondence between the reactants and the products, in particular we assume that product $p_j$ is the result of the transformation of reactant $r_j$ by the reaction.
In our experience, it is usually possible to translate a reaction of a pathway into such a ``normal form''. Reactions of cellular pathways very often represent bindings (and unbindings) of well-defined macromolecules, such as proteins and genes, to form (or to break) complexes either with other macromolecules or with small molecules such as ions and nutrients. Also conformational changes are common, in which a protein (or a complex constituted by a few proteins) changes its own ``state''. If we consider a complex not as a single entity, but as a combination of macromolecules we have that all of the mentioned kinds of reaction do not change the number of (macro)molecules in the system. Hence, it should be possible to model them with the reactions in the form we assume here. For the moment we leave the translation of reaction into the assumed form to the modeller.

In Section \ref{sec:syntax} we have introduced the syntax of the modelling formalism of biochemical pathways, in which a reaction is allowed to have a different number of reactants and products.

\subsection{Components identification}

Let us, thus, assume that the pathway $P$ consists of reactions in the following form:
\[
 \reactionC{r_1,\ldots,r_n}{p_1,\ldots,p_n}{c_1,\ldots,c_m}.
\]
Such a form enables us to identify a set of components $I$ that constitute the pathway.
Now we present an algorithm that given a pathway $P$ returns the set of components $I$ along with the partition of the set of species belonging to respective components.

We illustrate the intuitive idea on an example. Each reaction can be seen as a synchronisation of components. For example reaction $\reactionC{r_1,r_2}{p_1,p_2}{c}$ can be interpreted as a synchronisation of three components: 
one that changes its state from a state where $r_{1}$ holds into a state where $p_{1}$ is present and $r_{1}$ is not,
another component that changes its state from a state where $r_{2}$ holds to a state where $p_{2}$ is present and $r_{2}$ is not,
and a component which participates passively and stays in a state where $c$ is present.
Since we suppose that only one reaction takes place at a time in the whole system, the states of all the components do not change other than those involved in the reaction in the way we described.
From the example we can see that species $r_{1}$ and $p_{1}$ belong to the same component. Similarly $r_{2}$ belongs to the component that contains $p_{2}$, while $c$ is from a separate one.

\myskip

The algorithm follows. We start by assuming that each species belongs to a different component and we refine this assumption by iterating over the reactions constituting $P$. The result of the algorithm is a mapping $map$ assigning each species to its component.

\begin{algorithm}
\begin{algorithmic}
\STATE Let $map: S \mapsto I$ be an injective mapping
\FORALL{$R$ in $P$}
  \FORALL{$r_j$ in $\re(R)$}
  \STATE $map := \begin{cases}
              p \mapsto map(r_{j}) & \forall p \in \set{s\in S}{map(s)=map(p_j)}\\
              s \mapsto map(s) & \text{otherwise}
  \end{cases}$
  \ENDFOR
\ENDFOR
\RETURN $map$
\end{algorithmic}
\caption{Algorithm to partition species into different components}\label{alg:map}
\end{algorithm}

The algorithm updates the mapping by unifying the elements assigned to reactants and products in the same position in a reaction, and this is done for all reactions in the pathway.

The set of components $\comp(P)=I$ of pathway $P$ is the image of mapping $map$. Components of a reaction $R$ denoted, using the same notation, as $\comp(R)$ are defined as $\comp(\{R\})$.

\subsection{Initial state}

We adopt a semi-automatic heuristic procedure to find an initial state of the pathway. The idea is the following: for each species $s$ in $\species(P)$, if there is no reaction creating it (i.e.~if $s\not\in \bigcup_{R\in P} \pro(R)$) then in the initial state $s$ is present. This means that species that cannot be produced are assumed to be present in the initial state. Otherwise their presence in the model would not be meaningful. Subsequently, we resort again to the partitioning of species according to components to find other species to be inserted. In particular, we find those components containing no species present in the previous phase. These components must contain loops, hence we choose manually some of their species to insert. All other species are assumed absent.

\subsection{Visualisation of component interaction}

A component interaction graph can be drawn which visualises the components of a pathway and their interactions.
It is a directed graph in which vertices are system components (elements of $I$) and edges connect components that are involved together in a reaction. If two components are both involved as reactants (and consequently products), the edge connecting them will not be oriented (displayed as bidirectional). If one of the two is involved as reactant and the other as catalyst, then the edge will start from the vertex representing the latter to the vertex representing the former. There is no edge between vertices representing components involved in the same reactions only as catalysts.

\subsection{The model}

The model \pathway is made up of 143 species and 80 reactions.
It is in the correct form assumed in Section~\ref{sec:modver-assump} and no preprocessing is needed.
After performing the components identification procedure, 14 components are identified.
On Figure \ref{fig:intgraph} we can see the component interaction graph of $\pathway$.
Each node of the graph is labelled by the intuitive name of the component that we have chosen.

Visually, we can do some simple observations on the component interaction graph.
We can identify enzymes like Phosphatase1, Phosphatase2 and Phosphatase3. We can see the first part of the pathway corresponding to the EGF receptor and its interaction with effectors, and its connection to the MAP kinase cascade through the component RasGDP.

\begin{figure}[h]
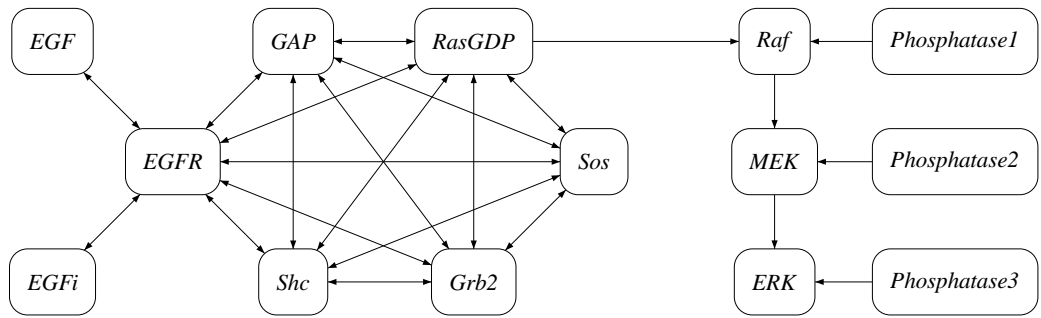

\begin{center}
\scalebox{\picscale}{
\begin{gpicture}(150,57)(0,-3)
\gasset{Nadjust=w,Nadjustdist=2,Nh=11,Nmr=3}

\node(A)(0,40){
  $\compfont{EGF}$ 
}
\node(B)(0,0){
  $\compfont{EGFi}$ 
}
\node(C)(20,20){
  $\compfont{EGFR}$ 
}
\node(D)(40,40){
  $\compfont{GAP}$ 
}
\node(E)(40,0){
  $\compfont{Shc}$ 
}
\node(F)(70,40){
  $\compfont{RasGDP}$ 
}
\node(G)(70,0){
  $\compfont{Grb2}$ 
}
\node(H)(90,20){
  $\compfont{Sos}$ 
}
\node(I)(150,40){
  $\compfont{Phosphatase1}$ 
}
\node(J)(120,40){
  $\compfont{Raf}$ 
}
\node(K)(120,20){
  $\compfont{MEK}$ 
}
\node(L)(120,0){
  $\compfont{ERK}$ 
}
\node(M)(150,20){
  $\compfont{Phosphatase2}$ 
}
\node(N)(150,0){
  $\compfont{Phosphatase3}$ 
}
\gasset{AHnb=1,ATnb=1}
\drawedge(A,C){}
\drawedge(B,C){}
\drawedge(C,D){}
\drawedge(C,E){}
\drawedge(C,F){}
\drawedge(C,G){}
\drawedge(C,H){}
\drawedge(D,E){}
\drawedge(D,F){}
\drawedge(D,G){}
\drawedge(D,H){}
\drawedge(E,F){}
\drawedge(E,G){}
\drawedge(E,H){}
\drawedge(F,G){}
\drawedge(F,H){}
\drawedge(G,H){}
\gasset{AHnb=1,ATnb=0}
\drawedge(F,J){}
\drawedge(I,J){}
\drawedge(J,K){}
\drawedge(K,L){}
\drawedge(M,K){}
\drawedge(N,L){}
\end{gpicture}
}
\end{center}
\caption{Component interaction graph of \pathway} 
\label{fig:intgraph}
\end{figure}

\section{Modular Verification}

In this section we define a modular verification technique for pathway models. 
We proceed by defining the projection of a pathway with help of the identified components.
Such a projection can be seen as an abstraction, giving rise to abstract pathways. 
We prove that a successful verification of a property in the abstraction implies its truth in the original model.

\subsection{Abstract pathways: syntax, semantics and fairness}

We are interested in analysing only a portion of the entire pathway, in particular a portion induced by only a subset of all components.
Let $I=\comp(P)$ and $J\subseteq I$, we define the projection of a pathway $P$ onto $J$ as an abstract pathway $P \proj J$.

We will need an extension of function \species, abusing the notation, which operates on a component set:  $\species(J)=\set{s\in \species(R)}{R \text{ s.t.~} \comp(R)\cap J\neq \emptyset}$.

\begin{definition}
 An abstract pathway $P\proj J$ is a pair $(PR,AR)$, where 
\begin{itemize}
 \item $PR=\set{R\in P}{\comp(R)\subseteq J}$
 \item $AR=\bigcup_{\substack{R\in P,\ \comp(R)\cap J\neq\emptyset\\  \comp(R)\cap(I\setminus J)\neq\emptyset}}
 \{\reactionC{\re(R)\proj J}{\pro(R)\proj J}{\cat(R)\proj J},\ \reactionC{\re(R)\proj J}{\re(R)\proj J}{\cat(R)\proj J}\}$
\end{itemize}
where the projection of a set of species $u\subseteq S$ is defined as  $u\proj J=\set{s\in u}{s\in \species(J)}$.
\end{definition}

An abstract pathway consists of two sets of of reactions: $PR$ contains reactions which influence only components inside $J$, and $AR$ contains projections of reactions that influence components both inside and outside $J$. Reactions in $PR$ are exactly as in $P$, since all of the species involved in such reactions are considered in the abstract pathway. On the other hand, reactions in $AR$ are obtained from the reactions in $P$ which involve species both from some components in $J$ and from some components not in $J$. For each of such reactions in $P$ we have two reactions in $AP$: one describing the situation in which species not in $species(J)$ are assumed to be configured such that the reaction can occur, and the other describing the opposite situation. In the first case the reaction in $AR$ produces some products; in the second case the reaction in $AR$ performs a self-loop (i.e. it does not change the state).

The abstract pathway semantics is defined as $\Semabs: P\proj J \mapsto \Sem(PR \cup AR)$, that is by using the standard semantics \Sem on the projected reactions in both $PR$ and $AR$.

What changes with respect to the pathway model is the definition of fairness. In fact, in this case fairness constraints can be applied only to reactions in $PR$ since reactions in $AR$ consist of pairs of reactions always applicable at the same time and in which it is reasonable to assume that one of the two is always preferred (describing the situation in which the corresponding reaction in $R$ is always enabled or disabled).
We define the notion of abstract fairness as 
$$\fairabs \iff \bigwedge_{R\in PR} (\GF\  enabled(R)\limpl \GF\ \occR).$$
Here the compassion is only required for reactions from $PR$.

Note that a pathway is a special case of abstract pathway, since the semantics of $P$ is equivalent (isomophic) to that of $P\proj I$ and in this case also $\fair \equiv \fairabs$ holds.

\subsection{Logic for specifying properties}

Properties of pathways are specified in temporal logic with species as atomic propositions.

The logic we consider is a fragment of the Computation Tree Logic CTL. Following Attie and Emerson \cite{AE98}, we assume the logic \ACTL for specification of properties. ACTL is the ``universal fragment'' of CTL which results from CTL by restricting negation to propositions and eliminating the existential path quantifier and \ACTL is ACTL without the AX modality.

\begin{definition}
The \emph{syntax of \ACTL} is defined inductively as follows:
\begin{itemize}
\item The constants $\true$ and $\false$ are formulae. $s$ and $\non s$ are formulae for any atomic proposition $s$, where the set of atomic propositions $AP$ are the set of all species $S$. 
\item If $f, g$ are formulae, then so are $f \land g$ and $f \lor g$.
\item If $f, g$ are formulae, then so are $\AU{f}{g}$ and $\AUw{f}{g}$.
\end{itemize}
\end{definition}

We define the logic ACTL$^{-}_{J}$ to be \ACTL where the atomic propositions are drawn from $\AP_{J}=\species(J)$. Abbreviations in \ACTL: $\mi{AF f} \equiv \AU{\true}{f}$ and $\mi{AG\ f} \equiv \AUw{f}{\false}$.


\myskip

Properties expressible by ACTL$^{-}$ formulae represent a significant class of properties investigated in the systems biology literature as identified in \cite{Monteiro:2008ft}, such as properties concerning exclusion (\emph{``It is not possible for a state $\state{}$ to occur''}), necessary consequence (\emph{``If a state $\state{1}$ occurs, then it is necessarily followed by a state $\state{2}$}''), and necessary persistence (\emph{``A state $\state{}$ must persist indefinitely''}).

On the other hand, properties as occurrence, possible consequence, sequence and possible persistence are of inherently existential nature, and are not expressible in ACTL$^{-}$.

\myskip

Definition of the semantics of \ACTL formulae on labelled transition system $\Sem(P)$ follows. Note that only fair maximal paths are considered.

\begin{definition}\emph{Semantics of \ACTL}. We define $\Sem(P),\state{} \satfair f$ (resp. $\Sem(P),\pi\satfair f$) meaning that $f$ is true in structure $\Sem(P)$ at state $\state{}$ (resp.~fair maximal path $\pi$). We define $\satfair$ inductively:

\begin{itemize}
\item $\Sem(P),\state{} \satfair \true$. $\Sem(P),\state{} \notsat \false$. $\Sem(P),\state{} \satfair s \textrm{ iff } \state{}(s)=\mi{tt}$. $\Sem(P),\state{} \satfair \neg s \textrm{ iff } \state{}(s)=\mi{ff}$.
\item $\Sem(P),\state{} \satfair f \land g \textrm{ iff } \Sem(P),\state{} \satfair f \textrm{ and } \Sem(P),\state{} \satfair g$.\\ $\Sem(P),\state{} \satfair f \lor g \textrm{ iff } \Sem(P),\state{} \satfair f \textrm{ or } \Sem(P),\state{} \satfair g$.
\item $\Sem(P),\state{} \satfair Af \textrm{ iff for every fair maximal path } \pi=(\state{},R,\hdots) \textrm{ in } \Sem(P): \Sem(P),\pi\satfair f$.
\item $\Sem(P),\pi \satfair f \textrm{ iff } \Sem(P),\state{} \satfair f$, where $\state{}$ is the first state of $\pi$
\item $\Sem(P),\pi \satfair f \land g \textrm{ iff } \Sem(P),\pi\satfair f \textrm{ and } \Sem(P),\pi\satfair g$.\\ $\Sem(P),\pi \satfair f \lor g \textrm{ iff } \Sem(P),\pi\satfair f \textrm{ or } \Sem(P),\pi\satfair g$.
\item $\Sem(P),\pi \satfair f\ U\ g \text{ iff } \pi = (\state{0},R_0,\state{1},R_1,\hdots) 
\text{ and there is } m \in\Nat \textrm{ such that } \Sem(P),\state{m} \satfair g \\ \textrm{ and for all } m'<m: \Sem(P),\state{m'}\satfair f$.
\item $\Sem(P),\pi \satfair f\ U_{w}\ g \text{ iff } \pi = (\state{0},R_0,\state{1},R_1,\hdots) 
\text{ and for all } m \in \Nat, \textrm{ if } \Sem(P),\state{m'}\notsat g \\ \textrm{ for all } m'<m \textrm{ then } \Sem(P),\state{m}\satfair f$.
\end{itemize}
\end{definition}

We assume $\satfairabs$ to be defined as $\satfair$, but with abstract fairness $\fairabs$ replacing $\fair$.

\subsection{Modular verification theorems}


Now we prove that in order to verify an $\ACTL_{J}$ property for a pathway $P$, it is enough to verify the same property in the abstract semantics of the abstract pathway $P\proj J$. The principle behind property preservation is that each path in the semantics of the modelled pathway must have a corresponding abstract path in the abstract semantics of a model obtained by projection. This, combined with the fact that \ACTL properties are universally quantified (namely describe properties that have to be satisfied by all paths) ensure that if an \ACTL property holds in the abstract semantics of the projection, then it will also hold in the semantics of the orginal model. In fact, for the components considered in a projection the semantics of the original model will contain essentially a subset of the paths of the projected model.

\ACTL properties are universally quantified, namely they they deal with all paths starting form a given initial state, and the fact that all original paths (more precisely their projections) are included amongst the paths of the projection. Thus if one proves that the property holds in the projection for all paths it will hold for all paths also in the original system.

First we define the path projection, which from a path in semantics of a pathway with the set of components $I$ removes transitions made by components outside of portion $J\subseteq I$ and restricts the rest of transitions onto $J$.

\begin{definition}
  \[
  \pi \sproj J = 
  \begin{cases}
    \emptypath & \text{ if } \pi=\emptypath\\
    \pi'\sproj J & \text{ if } \pi=\state{},R,\pi' \text{ and } \comp(R)\cap J = \emptyset\\
    \state{}\sproj J,R,\pi'\sproj J & \text{ if } \pi=\state{},R,\pi' \text{ and } \comp(R)\cap J \neq \emptyset\\
  \end{cases}
  \]
\end{definition}

Follows the infinite path projection which ensures that the resulting traces are infinite. In case of a finite original trace it adds an infinite looping in the final state.

\begin{definition}
Given $\pi\sproj J$ with initial state $\state{0}$ we define $\pi\sprojinf J = \pi\sproj J$ if $\pi\sproj J$ is infinite, otherwise if 
$\pi \sproj J = \state{0},R_{0},\hdots,\state{n-1},R_{n-1},\state{n}$ we define
$\pi\sprojinf J = \pi \sproj J,(*,\state{n})^{\infty}$, 
where $(R,\state{})^{\infty}=R,\state{},(R,\state{})^{\infty}$.
We denote $\emptypath\sprojinf J=\state{0},(*,\state{0})^{\infty}$ as $\emptypathinf$.
\end{definition}

Now we are in the position to present the crucial result, which states that a fair maximal path in the semantics of a pathway is either projected or infinitely projected into an abstractly fair maximal path in the abstract semantics of an abstract pathway. It is split in two lemmas, where the first one states that at least one of the projections is present as a maximal path in the abstract semantics. The second lemma proves the abstract fairness of the projections.

\begin{lemma}
\label{lem:pathproj}
$\pi \in \Sem(P)$ with $\pi\satLTL \fair$ implies $(\pi \sproj J \in \Semabs(P\proj J)$ or $\pi \sprojinf  J \in \Semabs(P\proj J))$.
\end{lemma}
\begin{proof}
Let us assume that $\pi$ is a finite path, then we can prove that either $(\pi \sproj J \in \Semabs(P\proj J)$ or $\pi \sprojinf  J \in \Semabs(P\proj J))$ by induction on the path length.


Case $\pi = \emptypath$. We have that $\pi \sproj J = \pi \sprojinf J = \emptypath$.
Since $\pi = \emptypath$, the initial state $\state{0}$ is such that no reaction is enabled in $\state{0}$. By the definition of abstract pathway semantics we know that the initial state of $\Semabs(P\proj J)$ is $\state{0}\sproj J$.

By definition of abstract pathway we have that $P\proj J = (PR,AR)$. 
\begin{itemize}
 \item If $AR = \emptyset$, then since $PR \subseteq P$ and there is no reaction in $P$ that is enabled in $\state{0}$, there is also no reaction in $PR$ that is enabled in $\state{0}\sproj J$, and hence $\emptypath \in \Semabs(P\proj J)$.
 \item If $AR \neq \emptyset$, then for each $R \in P \setminus PR$ we have in $AR$ two reactions $R_1,R_2$ as follows
  \begin{center}
 $R_{1}= \reactionC{\re(R)\proj J}{\pro(R)\proj J}{\cat(R)\proj J}$\\
 $R_{2}= \reactionC{\re(R)\proj J}{\re(R)\proj J}{\cat(R)\proj J}$
 \end{center}
Note that if $R_{1}$ is enabled, then also $R_{2}$ is enabled. 

As before, since there is no reaction in $P$ that is enabled in $\state{0}$, there is also no reaction in $PR$ that is enabled in $\state{0}\sproj J$.
If there is some $R_{2}$ enabled in $\state{0}\sproj J$, then $(R_{2},\state{0}\sproj J)^\infty \in  \Semabs(P\proj J)$, that is $\emptypath \sprojinf J \in \Semabs(P\proj J))$. On the other hand, if there is no $R_{2}$ enabled in $\state{0}\sproj J$, then there is also no $R_{1}$ enabled in the same state, hence $\emptypath \in \Semabs(P\proj J)$, that is $\emptypath \sproj J \in \Semabs(P\proj J)$.
\end{itemize}
%
%
%
%

Case $\pi = \state{},R,\pi'$. 
We distinguish two subcases:
\begin{itemize}
 \item $\comp(R)\cap J = \emptyset$: $\state{}'$ is the initial state of $\pi'$, so by induction hypothesis $\state{}'\sproj J$ is the initial state of either $\pi'\sproj J$ or $\pi' \sprojinf J$. Moreover, state, by definition of path projection, $\state{}\sproj J=\state{}'\sproj J$ then $\state{}\sproj J$ is the initial state of either $\pi'\sproj J$ or $\pi'\sprojinf J$, which means that either $\pi'\sproj J = \pi\sproj J \in \Semabs(P\proj J)$ or $\pi'\sprojinf J = \pi\sprojinf J \in \Semabs(P\proj J)$.

 \item $\comp(R)\cap J \neq \emptyset$: $\state{}'$ is the initial state of $\pi'$, so by induction hypothesis $\state{}'\sproj J$ is the initial state of either $\pi'\sproj J$ or $\pi' \sprojinf J$. Moreover, in $P\proj J$ there are $R_{1}$ and $R_{2}$ as above. Since $R$ is enabled in $\state{}$, $R_{1}$ is enabled in $\state{}\sproj J$. Therefore either $\state{}\sproj J, R_{1}, \pi'\sproj J = \pi\sproj J \in \Semabs(P\proj J)$ or $\state{}\sproj J, R_{1}, \pi'\sprojinf J = \pi\sprojinf J \in \Semabs(P\proj J)$.
 \end{itemize}

Summarising, if $\pi$ is a finite path, in all the possible cases we have that either $\pi \sproj J \in \Semabs(P\proj J)$ or $\pi \sprojinf  J \in \Semabs(P\proj J)$. 

Let us consider now the case in which $\pi$ is infinite. We have two subcases:
\begin{itemize} 
 \item  $\pi \sproj J$ is finite: By the fairness assumption $\pi\satLTL \fair$ it follows that if the path $\pi \sproj J$ is finite, then its final state is such that no reaction in $PR$ is enabled. In fact, as a consequence of fairness, the only case in which $\pi \sproj J$ can be finite is when all of the reactions that are infinitely often enabled in $\pi$ are performed only by components that are not in $J$. Hence, $\pi \sproj J$ is finite since it contains reactions that are not enabled infinitely often. Let $\pi = \pi_1,\pi_2$, where $\pi_1$ is the shortest (finite) prefix such that $\pi_2$ contains only moves of reactions enabled infinitely often. It is easy to see that $\pi \sproj J = \pi_1 \sproj J$ and similarly as before we can prove by induction on the length of $\pi_1$ that either $\pi_1 \sproj J \in \Semabs(P\proj J)$ or $\pi_1 \sprojinf  J \in \Semabs(P\proj J)$. The fact that no reaction is enabled in the final state of $\pi \sproj J$ ensures that $\pi \sproj J$ is a maximal path of $\Semabs(
P\proj J)$.
 \item  $\pi \sproj J$ is infinite: This case can be proved by showing that the relation between $\pi$ and $\pi \sproj J$ described in the inductive case of the proof for a finite $\pi$ is indeed an invariant property in the case of an infinite path $\pi$.
\end{itemize}
\end{proof}

Now we state and prove the second lemma.

\begin{lemma}
\label{lem:fairproj}
$\pi\satLTL \fair$ implies $\pi\sproj J \satLTL \fairabs$ and $\pi\sprojinf J \satLTL \fairabs$.
\end{lemma}
\begin{proof}
Suppose that $\pi\satLTL \fair$, \ie $\pi \satLTL \bigwedge_{R\in P} (\GF\  enabled(R)\limpl \GF\ \occR)$.
Let $J\subseteq I$ and $P\proj J=(AR,PR)$.
We want to prove that 
$\pi\sproj J \satLTL \fairabs$, that is
$\pi\sproj J \satLTL \bigwedge_{R\in PR} (\GF\  enabled(R)\limpl \GF\ \occR)$
This holds because of two facts (1) and (2) that can be easily checked: for any reaction $R$ from $PR$
\begin{itemize}
 \item $\state{} \satLTL \enabled(R)$ implies $\state{} \sproj J \satLTL \enabled(R)$ \hfill (1)
 \item $\state{} \satLTL \occR$ implies $\state{} \sproj J \satLTL \occR$ \hfill (2)
\end{itemize}
Analogously $\pi\sprojinf J \satLTL \bigwedge_{R\in PR} (\GF\  enabled(R)\limpl \GF\ \occR)$. 
\end{proof}

Finally, the property preservation theorem states that a successful verification of a property in the abstraction implies its truth in the original model.

\begin{theorem}
\label{thm:proppres}
For a pathway $P$ and a $J\subseteq I$ where $I$ is the component set of $P$ and $f$ an \ACTL formula we have
$\Semabs(P\proj J)\satfairabs f$ implies $\Sem(P)\satfair f$.
\end{theorem}
\begin{proof}
By induction on the structure of $f$ (for all $\state{}$).

$f=s$. By definition of state projection and the fact that $\AP_{R}$s are pairwise disjoint, for all atomic propositions $s$ from $\AP_{J}$ we get that $\Semabs(P\proj J),\state{}\sproj J\satfairabs s$ iff $\Sem(P),\state{} \satfair s$. Analogously for $f=\neg s$.

$f=g\land h$. From the assumption $\Semabs(P\proj J),\state{}\sproj J \satfairabs g\land h$ by CTL semantics, $\Semabs(P\proj J),\state{}\sproj J \satfairabs g$ and $\Semabs(P\proj J),\state{}\sproj J \satfairabs h$. By induction hypothesis $\Sem(P),\state{} \satfair g$ and $\Sem(P),\state{} \satfair h$. Hence, $\Sem(P),\state{} \satfair g\land h$. Case $f=g\lor h$ is proved analogously.

$f=\AUw{g}{h}$. Let $\pi$ be an arbitrary fair maximal path starting in $\state{}$. We establish $\Sem(P),\pi \satfair [g\ U_{w}\ h]$. By Lemma \ref{lem:pathproj} at least one of  $\pi\sproj J$ or $\pi\sprojinf J$ is a path in $\Semabs(P\proj J)$, and by Lemma \ref{lem:fairproj} both are abstractly fair.

Let us suppose first that $\pi\sproj J$ is the abstractly fair maximal path in $\Semabs(P\proj J)$. Hence by the assumption $\Semabs(P\proj J),\pi\sproj J \satfairabs [g\ U_{w}\ h]$. There are two cases:
\begin{enumerate}
\item $\Semabs(P\proj J),\pi \sproj J \satfairabs G\ g$. Let $t$ be any state along $\pi$. By CTL semantics $\Semabs(P\proj J),t\sproj J\satfairabs g$. By induction hypothesis we have $\Sem(P),t\satfair g$. Since $t$ was an arbitrary state of $\pi$, we get $\Sem(P),\pi\satfair G\ g$ and thus $\Sem(P),\pi \satfair g\ U_{w}\ h$.
\item $\Semabs(P\proj J),\pi \sproj J \satfairabs [g\ U\ h]$. Let $\state{}_{J}^{m''}$ be the first state along $\pi\sproj J$ that satisfies $h$. Then there is at least one state $\state{}^{m''}$ along $\pi$ such that $\state{}^{m''}\sproj J=\state{}_{J}^{m''}$. Let $\state{}^{m'}$ be first such state. By induction hypothesis $\Sem(P),\state{}^{m'}\satfair h$. From the definition of path projection any $\state{}^{m}$ with $m<m'$ projects to $\state{}^{m}\sproj J$ that is before $\state{}_{J}^{m'}$ in $\pi\sproj J$. By the assumption $\Semabs(P\proj J),\state{}^{m}\sproj J\satfairabs g$, hence by induction hypothesis $\Sem(P),\state{}^{m}\satfair g$. By CTL semantics we get $\Sem(P),\pi \satfair g\ U\ h$.
\end{enumerate}
In both cases we showed $\Sem(P),\pi \satfair g\ U_{w}\ h$. Since $\pi$ was arbitrary fair maximal path starting in $\state{}$, we conclude $\Sem(P),\state{} \satfair \AUw{g}{h}$.

The reasoning for the case in which the abstractly fair maximal path in $\Semabs(P\proj J)$ is $\pi\sprojinf J$ is analogous to the considered case.

$f=\AU{g}{h}$. Let $\pi$ be an arbitrary fair maximal path starting in $\state{}$. By Lemmas \ref{lem:pathproj} and \ref{lem:fairproj} we have that $\pi\sproj J$ or $\pi\sprojinf J$ is a fair maximal path in $\Semabs(P\proj J)$ and by the assumption $\Semabs(P\proj J),\pi\sproj J \satfairabs [{g\ U\ h}]$ or $\Semabs(P\proj J),\pi\sprojinf J \satfairabs [{g\ U\ h}]$. By the above case we get $\Sem(P),\state{} \satfair \AU{g}{h}$.
\end{proof}
\section{Experiments}\label{sec:experiments}

In this section we exploit the NuSMV model checker to perform some experiments on the model of the EGF pathway. NuSMV includes model checking algorithms that allow fairness constraints to be taken into account. We rely on such algorithms to manage fairness constraints introduced in this paper. Moreover, in order to carry out the projection and encode the resulting abstract pathway in the NuSMV format we have developed a tool (available upon request).

The first experiment is aimed at showing how modular verification could be applied to verify a global property of the pathway, namely that the final product of the pathway is always produced. This can be done in a modular way by proving sub-properties in three different model fragments obtained by projection. 

Subsequently, a number of experiments are performed with the aim of showing how the molecular components we identified in the pathway can be used to better understand the pathway dynamics. In particular, we check whether there are some molecular components that are not really necessary to obtain the final product of the pathway. This will be done by applying model checking on models in which molecular components are selectively disabled by setting their initial states to false. Also in this case the modular verification approach is adopted.

In this case study modular verification allows properties to be verified faster than on the complete model. However, modular verification is still not significantly more efficient than verification on the complete model. This is due to the projection operation we are considering at the moment, which is rather rough. In Section~\ref{sec:conclusions} we discuss why this modular verification is a promising approach for the analysis of pathways, and how we plan to improve the approach to make it substantially more efficient.

To run the experiment we used NuSMV 2.5.4 on a workstation equipped with an Intel i5 CPU 2.80 Ghz, with 8GB RAM and running Ubuntu GNU/Linux. In order to make verification faster NuSMV was executed in batch mode by enabling dynamic reordering of BDD variables and by disabling the generation of counterexamples.

\subsection{Modular verification of a global property}\label{experiments-1}

The final product of the MAP kinase cascade activated by surface and internalised EGF receptors is species \ERKPP. Since surface and internalised receptors activate two different branches of the pathway, we denote by \ERKPP the product of the branch activated by the surface receptors and by \ERKPPi the product of the branch activated by the internalised receptors.

The property to be verified is
\begin{equation}\label{global-property}
AF (\ERKPP \lor \ERKPPi)
\end{equation}

The property holds in the complete model and its verification required 260 seconds. By looking at the diagram in Figure \ref{fig:egf-pathway} we noticed that the pathway could be partitioned in three parts, with two species acting as ``gates''. These two species are $\specI$ and $\Rafstar$. Hence, we decided to try to apply modular verification by splitting property \ref{global-property} into the following three sub-properties:
\begin{align}
& AF (\specI) \label{sub-property-1}\\
& AG (\specI \implies (AF\ \Rafstar)) \label{sub-property-2}\\
& AG (\Rafstar \implies AF (\ERKPP \lor \ERKPPi)) \label{sub-property-3}
\end{align}
Property (\ref{sub-property-1}) states that in all paths of the system a state in which species $\specI$ is present is eventually reached. Property (\ref{sub-property-2}) states that whenever a state is reached in which species $\specI$ is present, then a state in which $\Rafstar$ is present is eventually reached. Finally, property (\ref{sub-property-3}) states that whenever a state is reached in which in which species $\Rafstar$ is present, then a state in which either $\ERKPP$ or $\ERKPPi$ is present is eventually reached. It is easy to see that the conjunction of (\ref{sub-property-1}), (\ref{sub-property-2}) and (\ref{sub-property-3}) implies (\ref{global-property}).

We considered three projections of the complete model to be used to verify properties (\ref{sub-property-1}), (\ref{sub-property-2}) and (\ref{sub-property-3}), respectively. In particular, from the component interaction graph of the model (shown in Figure~\ref{fig:intgraph}) we extracted the following subsets to be used for projections:
\begin{itemize}
 \item in order to verify (\ref{sub-property-1}) we considered the subset $J_1$ consisting of components $\compfont{EGF}$, $\compfont{EGFi}$, $\compfont{EGFR}$ and $\compfont{GAP}$;
 \item in order to verify (\ref{sub-property-2}) we considered the subset $J_2$ consisting of components $\compfont{EGFR}$, $\compfont{GAP}$, $\compfont{Shc}$, $\compfont{RasGDP}$, $\compfont{Grb2}$ and $\compfont{Sos}$;
 \item in order to verify (\ref{sub-property-3}) we considered the subset $J_3$ consisting of components $\compfont{RasGDP}$, $\compfont{Raf}$, $\compfont{MEK}$, $\compfont{ERK}$, $\compfont{Phosphatase1}$, $\compfont{Phosphatase2}$ and $\compfont{Phosphatase3}$.
\end{itemize}
We obtained that (\ref{sub-property-1}), (\ref{sub-property-2}) and (\ref{sub-property-3}) hold in the abstract semantics of the abstract pathways $P \proj J_1$, $P \proj J_2$ and $P \proj J_3$, respectively. Moreover, model checking required less than three seconds for (\ref{sub-property-1}), 213 seconds for (\ref{sub-property-2}) and less than one second for (\ref{sub-property-3}). Overall, modular verification required 217 seconds, that is 43 seconds less than verification on the complete model.

\subsection{Reasoning on molecular components}\label{experiments-2}

As it can be seen in the component interaction graph and in the diagram in Figure \ref{fig:egf-pathway}, some molecular components are involved in complex interactions. This is true in particular for components $\compfont{EGFR}$, $\compfont{GAP}$, $\compfont{RasGDP}$, $\compfont{Sos}$, $\compfont{Shc}$ and $\compfont{Grb2}$ which form a clique in the component interaction graph.
We are interested in understanding whether all of these components are really necessary in order to obtain the final products of the pathway.
The idea is to test whether the final species are produced when the components of interest are assumed one by one as disabled.
Molecular components $\compfont{EGFR}$ and $\compfont{RasGDP}$ are for sure necessary since they connect the clique with the other molecular components of the pathway. Consequently, we focus our analysis on $\compfont{GAP}$, $\compfont{Sos}$, $\compfont{Shc}$ and $\compfont{Grb2}$.

In order to disable a molecular component we consider as absent all of its species in the initial state of the systems. Hence, we consider a set of four (complete) models, each with one of the four components under study disabled. On each model we try to verify property (\ref{global-property}): if the property does not hold, then the component that is disabled in such a model is necessary for the pathway; on the other hand, if the property holds, then the component turns out to be not necessary since the products of the pathway can be obtained even without it. The same tests can be also done in a modular way by decomposing the pathway and the property as in Section~\ref{experiments-1}.

\begin{table}[t]
\begin{center}
\begin{tabular}{|c|c|c|c|c|c|c|c|}
\hline 
\textbf{} & \multicolumn{3}{|c|}{\textbf{Verification complete model}} &  \multicolumn{4}{|c|}{\textbf{Modular Verification}}\\
\hline
\textbf{Disabled component} & \textbf{Property} & \textbf{Result} & \textbf{Time} & \textbf{Property} & \textbf{Result} & \textbf{Time} & \textbf{Total time}\\
\hline
\multirow{3}{*}{none}& \multirow{3}{*}{(\ref{global-property})}& \multirow{3}{*}{true} & \multirow{3}{*}{260s} & (\ref{sub-property-1}) & true & 3s & \\
& & & & (\ref{sub-property-2}) & true & 213s & 217s\\
& & & & (\ref{sub-property-3}) & true & 1s & \\
\hline
$\compfont{GAP}$& (\ref{global-property})& false & 252s & (\ref{sub-property-1}),(\ref{prop-strong1}) & false,true & 2s & 2s\\
\hline
\multirow{2}{*}{$\compfont{Sos}$}& \multirow{2}{*}{(\ref{global-property})}& \multirow{2}{*}{false} & \multirow{2}{*}{253s} &(\ref{sub-property-1}) & true & 3s & \multirow{2}{*}{210s}\\
& & & & (\ref{sub-property-2}),(\ref{prop-strong2}) & false,true & 207s & \\
\hline
\multirow{3}{*}{$\compfont{Shc}$}& \multirow{3}{*}{(\ref{global-property})}& \multirow{3}{*}{true} & \multirow{3}{*}{252s} & (\ref{sub-property-1}) & true & 3s & \multirow{3}{*}{212s}\\
& & & & (\ref{sub-property-2}) & true & 208s &\\
& & & & (\ref{sub-property-3}) & true & 1s &\\
\hline
\multirow{2}{*}{$\compfont{Grb2}$}& \multirow{2}{*}{(\ref{global-property})}& \multirow{2}{*}{false} & \multirow{2}{*}{253s} & (\ref{sub-property-1}) & true & 3s & \multirow{2}{*}{211s}\\
& & & & (\ref{sub-property-2}),(\ref{prop-strong2}) & false,true & 208s &\\
\hline
\end{tabular}
\end{center}
\caption{Model checking results and comparison of verification times}\label{tab:comparisonpathway}
\end{table}

In Table \ref{tab:comparisonpathway} we summarise the property verification results and compare verification times obtained by model checking the complete models and by following the modular approach. The first row of data in the table reports verification results in which no component is disabled (as in Section~\ref{experiments-1}). The other results show that $\compfont{Shc}$ is not a necessary component, whereas all of the other three are. As previously, the time required by modular verification is smaller than the one required by model checking the complete model. This is true in particular in the case in which $\compfont{GAP}$ is disabled since property (\ref{sub-property-1}), the verification of which is very fast, turns out to be false.

Note that in the case of modular verification of the models in which $\compfont{GAP}$, $\compfont{Sos}$ and $\compfont{Grb2}$ were disabled we needed to verify some additional properties. In particular, in the case of $\compfont{GAP}$ we have that property (\ref{sub-property-1}) does not hold in the abstract semantics of $P\proj J_1$, and in the cases of $\compfont{Sos}$ and $\compfont{Grb2}$ property (\ref{sub-property-2}) does not hold in the abstract semantics of $P \proj J_2$. We remark that our modular verification approach guarantees only that properties proved to hold in a model fragment also hold in the complete model. Nothing can be said, instead, of properties that does not hold in the model fragments. In order to avoid applying model checking on the complete model to check whether these properties hold there, we consider some new properties whose satisfaction in suitable model fragments implies that properties (\ref{sub-property-1}) and (\ref{sub-property-2}) actually do not hold.
In order to prove that (\ref{sub-property-1}) is actually false when $\compfont{GAP}$ is disabled we consider the following property:
\begin{equation}\label{prop-strong1}
 AG (\neg \specI)
\end{equation}
In order to prove that (\ref{sub-property-2}) is actually false when either $\compfont{Sos}$ or $\compfont{Grb2}$ is disabled we consider the following property:
\begin{equation}\label{prop-strong2}
 AG (\neg\Rafstar)
\end{equation}

Note that it is convenient to verify properties (\ref{prop-strong1}) and (\ref{prop-strong2}) together with (\ref{sub-property-1}) and (\ref{sub-property-2}), respectively. This avoids spending twice the time needed by the model checker to construct the data structure necessary to perform the verification. In the case of our experiments the construction of such data structures takes usually the 98\%-99\% of the verification time. Times reported in Table~\ref{tab:comparisonpathway} are based on this optimisation.

\section{Discussion and conclusions}
\label{sec:conclusions}

In this paper we presented preliminary results in the development of a modular verification framework for biochemical pathways.
We defined a modelling notation for pathways associated with a formal semantics and a notion of fairness that allows the dynamics to be accurately described by avoiding starvation situations among reactions. Moreover, we investigated a notion of molecular component of a pathway and we provided a methodology to infer molecular components from pathways the reactions of which satisfy some assumptions. Molecular components were then used by a projection operation that allows abstract pathways modelling an over-approximation of the behaviour of a group of components to be obtained from a pathway model. 
The fact that a property expressed by means of the \ACTL logic holds in an abstract pathway was shown to imply that they hold also in the complete pathway model. This preservation is at the basis of the modular verification approach which was demonstrated on a well-established model of the EGF pathway.

The results of experiments given in Section~\ref{sec:experiments} show that our modular verification approach allows properties to be verified in a shorter time than in the case of verification of the complete pathway model. However, in most of the cases the time saved was relatively small ($\sim 15\%$). We believe that the cause of this limited gain in efficiency is due to the projection operation we are considering at the moment, which is still somewhat rough. Our plan to improve efficiency is to define a projection operation that combines the current one (that essentially removes some molecular components from the model) with another that somehow minimises the description of components not removed by the model, but whose role in the property to be verified is marginal. In the case of the considered case study this would allow, for example, to reduce the size of the model of the components constituting the clique in the component interaction graph in Figure~\ref{fig:intgraph} by focusing on 
components $\compfont{EGFR}$ and $\compfont{RasGDP}$, and by minimising the description of components $\compfont{GAP}$, $\compfont{Shc}$, $\compfont{Sos}$ and $\compfont{Grb2}$. This would allow for a significant improvement in modular verification efficiency.

\bibliographystyle{eptcs}
\bibliography{refs_new}

\begin{thebibliography}{10}
\providecommand{\bibitemdeclare}[2]{}
\providecommand{\surnamestart}{}
\providecommand{\surnameend}{}
\providecommand{\urlprefix}{Available at }
\providecommand{\url}[1]{\texttt{#1}}
\providecommand{\href}[2]{\texttt{#2}}
\providecommand{\urlalt}[2]{\href{#1}{#2}}
\providecommand{\doi}[1]{doi:\urlalt{http://dx.doi.org/#1}{#1}}
\providecommand{\bibinfo}[2]{#2}

\bibitemdeclare{article}{AE98}
\bibitem{AE98}
\bibinfo{author}{Paul~C. \surnamestart Attie\surnameend} \&
  \bibinfo{author}{E.~Allen \surnamestart Emerson\surnameend}
  (\bibinfo{year}{1998}): \emph{\bibinfo{title}{Synthesis of concurrent systems
  with many similar processes}}.
\newblock {\sl \bibinfo{journal}{ACM Transactions on Programming Languages and
  Systems}} \bibinfo{volume}{20}(\bibinfo{number}{1}), pp.
  \bibinfo{pages}{51--115}, \doi{10.1145/271510.271519}.

\bibitemdeclare{article}{1231124}
\bibitem{1231124}
\bibinfo{author}{Roberto \surnamestart Barbuti\surnameend},
  \bibinfo{author}{Andrea \surnamestart Maggiolo-Schettini\surnameend},
  \bibinfo{author}{Paolo \surnamestart Milazzo\surnameend} \&
  \bibinfo{author}{Angelo \surnamestart Troina\surnameend}
  (\bibinfo{year}{2006}): \emph{\bibinfo{title}{A Calculus of Looping Sequences
  for Modelling Microbiological Systems}}.
\newblock {\sl \bibinfo{journal}{Fundamenta Informaticae}}
  \bibinfo{volume}{72}(\bibinfo{number}{1-3}), pp. \bibinfo{pages}{21--35}.

\bibitemdeclare{article}{DBLP:journals/iandc/BurchCMDH92}
\bibitem{DBLP:journals/iandc/BurchCMDH92}
\bibinfo{author}{Jerry~R. \surnamestart Burch\surnameend},
  \bibinfo{author}{Edmund~M. \surnamestart Clarke\surnameend},
  \bibinfo{author}{Kenneth~L. \surnamestart McMillan\surnameend},
  \bibinfo{author}{David~L. \surnamestart Dill\surnameend} \&
  \bibinfo{author}{L.~J. \surnamestart Hwang\surnameend}
  (\bibinfo{year}{1992}): \emph{\bibinfo{title}{Symbolic Model Checking:
  $10^{20}$ States and Beyond}}.
\newblock {\sl \bibinfo{journal}{Information and Computation}}
  \bibinfo{volume}{98}(\bibinfo{number}{2}), pp. \bibinfo{pages}{142--170},
  \doi{10.1016/0890-5401(92)90017-A}.

\bibitemdeclare{article}{Cardelli:2005fk}
\bibitem{Cardelli:2005fk}
\bibinfo{author}{Luca \surnamestart Cardelli\surnameend}
  (\bibinfo{year}{2005}): \emph{\bibinfo{title}{Brane Calculi}}.
\newblock {\sl \bibinfo{journal}{Computational Methods in Systems Biology}},
  pp. \bibinfo{pages}{257--278}, \doi{10.1007/978-3-540-25974-9_24}.

\bibitemdeclare{inproceedings}{CAV02}
\bibitem{CAV02}
\bibinfo{author}{Alessandro \surnamestart Cimatti\surnameend},
  \bibinfo{author}{Edmund \surnamestart Clarke\surnameend},
  \bibinfo{author}{Enrico \surnamestart Giunchiglia\surnameend},
  \bibinfo{author}{Fausto \surnamestart Giunchiglia\surnameend},
  \bibinfo{author}{Marco \surnamestart Pistore\surnameend},
  \bibinfo{author}{Marco \surnamestart Roveri\surnameend},
  \bibinfo{author}{Roberto \surnamestart Sebastiani\surnameend} \&
  \bibinfo{author}{Armando \surnamestart Tacchella\surnameend}
  (\bibinfo{year}{2002}): \emph{\bibinfo{title}{{NuSMV} Version 2: An
  OpenSource Tool for Symbolic Model Checking}}.
\newblock In: {\sl \bibinfo{booktitle}{Proc. International Conference on
  Computer-Aided Verification (CAV 2002)}}, {\sl \bibinfo{series}{LNCS}}
  \bibinfo{volume}{2404}, \bibinfo{publisher}{Springer},
  \bibinfo{address}{Copenhagen, Denmark}, pp. \bibinfo{pages}{241--268},
  \doi{10.1007/3-540-45657-0_29}.

\bibitemdeclare{article}{1570750}
\bibitem{1570750}
\bibinfo{author}{Federica \surnamestart Ciocchetta\surnameend} \&
  \bibinfo{author}{Jane \surnamestart Hillston\surnameend}
  (\bibinfo{year}{2009}): \emph{\bibinfo{title}{Bio-{PEPA}: A framework for the
  modelling and analysis of biological systems}}.
\newblock {\sl \bibinfo{journal}{Theoretical Computer Science}}
  \bibinfo{volume}{410}(\bibinfo{number}{33-34}), pp.
  \bibinfo{pages}{3065--3084}, \doi{10.1016/j.tcs.2009.02.037}.

\bibitemdeclare{article}{186051}
\bibitem{186051}
\bibinfo{author}{Edmund~M. \surnamestart Clarke\surnameend},
  \bibinfo{author}{Orna \surnamestart Grumberg\surnameend} \&
  \bibinfo{author}{David~E. \surnamestart Long\surnameend}
  (\bibinfo{year}{1994}): \emph{\bibinfo{title}{Model checking and
  abstraction}}.
\newblock {\sl \bibinfo{journal}{ACM Transactions on Programming Languages and
  Systems}} \bibinfo{volume}{16}(\bibinfo{number}{5}), pp.
  \bibinfo{pages}{1512--1542}, \doi{10.1145/186025.186051}.

\bibitemdeclare{book}{Clarke99}
\bibitem{Clarke99}
\bibinfo{author}{Edmund~M. \surnamestart Clarke\surnameend},
  \bibinfo{author}{Orna \surnamestart Grumberg\surnameend} \&
  \bibinfo{author}{Doron \surnamestart Peled\surnameend}
  (\bibinfo{year}{1999}): \emph{\bibinfo{title}{Model Checking}}.
\newblock \bibinfo{publisher}{MIT Press}.

\bibitemdeclare{article}{1041036}
\bibitem{1041036}
\bibinfo{author}{Vincent \surnamestart Danos\surnameend} \&
  \bibinfo{author}{Cosimo \surnamestart Laneve\surnameend}
  (\bibinfo{year}{2004}): \emph{\bibinfo{title}{Formal molecular biology}}.
\newblock {\sl \bibinfo{journal}{Theoretical Computer Science}}
  \bibinfo{volume}{325}(\bibinfo{number}{1}), pp. \bibinfo{pages}{69--110},
  \doi{10.1016/j.tcs.2004.03.065}.

\bibitemdeclare{inproceedings}{DrabikNCMA10}
\bibitem{DrabikNCMA10}
\bibinfo{author}{Peter \surnamestart Dr{\'a}bik\surnameend},
  \bibinfo{author}{Andrea \surnamestart Maggiolo-Schettini\surnameend} \&
  \bibinfo{author}{Paolo \surnamestart Milazzo\surnameend}
  (\bibinfo{year}{2010}): \emph{\bibinfo{title}{Dynamic Sync-programs for
  Modular Verification of Biological Systems}}.
\newblock In: {\sl \bibinfo{booktitle}{2nd Int. Workshop on Non-Classical
  Models of Automata and applications (NCMA'10)}}, \bibinfo{volume}{263},
  \bibinfo{publisher}{Austrian Computer Society}, \bibinfo{address}{Jena,
  Germany}, pp. \bibinfo{pages}{71--83}.

\bibitemdeclare{article}{DrabikENTCS10}
\bibitem{DrabikENTCS10}
\bibinfo{author}{Peter \surnamestart Dr{\'a}bik\surnameend},
  \bibinfo{author}{Andrea \surnamestart Maggiolo-Schettini\surnameend} \&
  \bibinfo{author}{Paolo \surnamestart Milazzo\surnameend}
  (\bibinfo{year}{2010}): \emph{\bibinfo{title}{Modular Verification of
  Interactive Systems with an Application to Biology}}.
\newblock {\sl \bibinfo{journal}{Electronic Notes in Theoretical Computer
  Science}} \bibinfo{volume}{268}, pp. \bibinfo{pages}{61--75},
  \doi{10.1016/j.entcs.2010.12.006}.

\bibitemdeclare{article}{DrabikSACS11}
\bibitem{DrabikSACS11}
\bibinfo{author}{Peter \surnamestart Dr\'{a}bik\surnameend},
  \bibinfo{author}{Andrea \surnamestart Maggiolo-Schettini\surnameend} \&
  \bibinfo{author}{Paolo \surnamestart Milazzo\surnameend}
  (\bibinfo{year}{2011}): \emph{\bibinfo{title}{Modular Verification of
  Interactive Systems with an Application to Biology}}.
\newblock {\sl \bibinfo{journal}{Scientific Annals of Computer Science}}
  \bibinfo{volume}{21}, pp. \bibinfo{pages}{39--72}.

\bibitemdeclare{article}{EL87}
\bibitem{EL87}
\bibinfo{author}{E.~Allen \surnamestart Emerson\surnameend} \&
  \bibinfo{author}{Chin-Laung \surnamestart Lei\surnameend}
  (\bibinfo{year}{1987}): \emph{\bibinfo{title}{Modalities for model checking:
  branching time logic strikes back}}.
\newblock {\sl \bibinfo{journal}{Science of Computer Programming}}
  \bibinfo{volume}{8}, pp. \bibinfo{pages}{275--306},
  \doi{10.1016/0167-6423(87)90036-0}.

\bibitemdeclare{article}{Fages04modellingand}
\bibitem{Fages04modellingand}
\bibinfo{author}{Fran{\c c}ois \surnamestart Fages\surnameend},
  \bibinfo{author}{Sylvain \surnamestart Soliman\surnameend} \&
  \bibinfo{author}{Nathalie \surnamestart Chabrier-Rivier\surnameend}
  (\bibinfo{year}{2004}): \emph{\bibinfo{title}{Modelling and querying
  interaction networks in the biochemical abstract machine biocham}}.
\newblock {\sl \bibinfo{journal}{Journal of Biological Physics and Chemistry}}
  \bibinfo{volume}{4}, pp. \bibinfo{pages}{64--73}.

\bibitemdeclare{article}{1342619}
\bibitem{1342619}
\bibinfo{author}{John \surnamestart Heath\surnameend}, \bibinfo{author}{Marta
  \surnamestart Kwiatkowska\surnameend}, \bibinfo{author}{Gethin \surnamestart
  Norman\surnameend}, \bibinfo{author}{David \surnamestart Parker\surnameend}
  \& \bibinfo{author}{Oksana \surnamestart Tymchyshyn\surnameend}
  (\bibinfo{year}{2008}): \emph{\bibinfo{title}{Probabilistic model checking of
  complex biological pathways}}.
\newblock {\sl \bibinfo{journal}{Theoretical Computer Science}}
  \bibinfo{volume}{391}(\bibinfo{number}{3}), pp. \bibinfo{pages}{239--257},
  \doi{10.1016/j.tcs.2007.11.013}.

\bibitemdeclare{article}{Monteiro:2008ft}
\bibitem{Monteiro:2008ft}
\bibinfo{author}{Pedro~T. \surnamestart Monteiro\surnameend},
  \bibinfo{author}{Delphine \surnamestart Ropers\surnameend},
  \bibinfo{author}{Radu \surnamestart Mateescu\surnameend},
  \bibinfo{author}{Ana~T. \surnamestart Freitas\surnameend} \&
  \bibinfo{author}{Hidde \surnamestart de~Jong\surnameend}
  (\bibinfo{year}{2008}): \emph{\bibinfo{title}{Temporal logic patterns for
  querying dynamic models of cellular interaction networks}}.
\newblock {\sl \bibinfo{journal}{Bioinformatics}}
  \bibinfo{volume}{24}(\bibinfo{number}{16}), pp. \bibinfo{pages}{227--233},
  \doi{10.1093/bioinformatics/btn275}.

\bibitemdeclare{article}{ltl}
\bibitem{ltl}
\bibinfo{author}{Amir \surnamestart Pnueli\surnameend} (\bibinfo{year}{1981}):
  \emph{\bibinfo{title}{The temporal semantics of concurrent programs}}.
\newblock {\sl \bibinfo{journal}{Theoretical Computer Science}}
  \bibinfo{volume}{13}(\bibinfo{number}{1}), pp. \bibinfo{pages}{45 -- 60},
  \doi{10.1016/0304-3975(81)90110-9}.

\bibitemdeclare{article}{513277}
\bibitem{513277}
\bibinfo{author}{Corrado \surnamestart Priami\surnameend},
  \bibinfo{author}{Aviv \surnamestart Regev\surnameend}, \bibinfo{author}{Ehud
  \surnamestart Shapiro\surnameend} \& \bibinfo{author}{William \surnamestart
  Silverman\surnameend} (\bibinfo{year}{2001}):
  \emph{\bibinfo{title}{Application of a stochastic name-passing calculus to
  representation and simulation of molecular processes}}.
\newblock {\sl \bibinfo{journal}{Information Processing Letters}}
  \bibinfo{volume}{80}(\bibinfo{number}{1}), pp. \bibinfo{pages}{25--31},
  \doi{10.1016/S0020-0190(01)00214-9}.

\bibitemdeclare{article}{citeulike:1180143}
\bibitem{citeulike:1180143}
\bibinfo{author}{Aviv \surnamestart Regev\surnameend},
  \bibinfo{author}{Ekaterina~M. \surnamestart Panina\surnameend},
  \bibinfo{author}{William \surnamestart Silverman\surnameend},
  \bibinfo{author}{Luca \surnamestart Cardelli\surnameend} \&
  \bibinfo{author}{Ehud \surnamestart Shapiro\surnameend}
  (\bibinfo{year}{2004}): \emph{\bibinfo{title}{BioAmbients: an abstraction for
  biological compartments}}.
\newblock {\sl \bibinfo{journal}{Theoretical Computer Science}}
  \bibinfo{volume}{325}(\bibinfo{number}{1}), pp. \bibinfo{pages}{141--167},
  \doi{10.1016/j.tcs.2004.03.061}.

\bibitemdeclare{article}{egf-reduced}
\bibitem{egf-reduced}
\bibinfo{author}{Birgit \surnamestart Schoeberl\surnameend},
  \bibinfo{author}{Claudia \surnamestart Eichler-Jonsson\surnameend},
  \bibinfo{author}{Ernst~Dieter \surnamestart Gilles\surnameend} \&
  \bibinfo{author}{Gertraud \surnamestart Muller\surnameend}
  (\bibinfo{year}{2002}): \emph{\bibinfo{title}{Computational modeling of the
  dynamics of the MAP kinase cascade activated by surface and internalized EGF
  receptors}}.
\newblock {\sl \bibinfo{journal}{Nature Biotechnology}}
  \bibinfo{volume}{20}(\bibinfo{number}{4}), pp. \bibinfo{pages}{370--375},
  \doi{10.1038/nbt0402-370}.

\end{thebibliography}


\end{document}